\documentclass[prodmode,permissions]{acmsmall-ec16}
\acmVolume{X}
\acmNumber{X}
\acmArticle{X}
\acmYear{2016}
\acmMonth{2}

\usepackage{times}
\usepackage{graphicx} 
\usepackage{subfigure} 

\usepackage[numbers,sort&compress]{natbib} 

\usepackage{latexsym}
\usepackage{amsmath,amssymb,amsfonts}
\usepackage{graphicx} 
\usepackage{float,url, subfigure, enumitem, color}

\usepackage[utf8]{inputenc}
\usepackage[english]{babel}

\usepackage{mathtools}

\DeclarePairedDelimiter\floor{\big\lfloor}{\big\rfloor}

\usepackage{nicefrac}
\newcommand{\xfrac}[3][]{\nicefrac[#1]{#2}{#3}}

\newtheorem{thm}{Theorem}
\newtheorem{lem}[thm]{Lemma}
\newtheorem{prop}[thm]{Proposition}
\newtheorem{cor}[thm]{Corollary}
\newtheorem{defn}[thm]{Definition}
\newtheorem{rem}[thm]{Remark}

\newcommand{\R}{\mathbb{R}}

\begin{document}

\title{A Model for Social Network Formation}

\markboth{Celis and Mousavifar}{Network Formation}

\title{A Model for Social Network Formation: \\
Efficiency, Stability and Dynamics}

\author{L. Elisa Celis \affil{\'Ecole Polytechnique F\'ed\'erale de Lausanne}  
	Aida Sadat Mousavifar \affil{University of Tehran}}

\begin{abstract}

We introduce a simple network formation model for social networks. Agents are nodes, connecting to another agent by building a directed edge (or accepting a connection from another agent) has a cost, and reaching (or being reached by) other agents via short directed paths has a benefit; in effect, an agent wants to reach others quickly, but without the cost of directly connecting each and every one. We prove that asynchronous edge dynamics always converge to a stable network; in fact, for nontrivial ranges of parameters this convergence is fast. Moreover, the set of fixed points of the dynamics form a nontrivial class of networks. For the static game, we give classes of efficient networks for nontrivial parameter ranges and further study their stability. We close several problems, and leave many interesting ones open.

\end{abstract}



\maketitle

\section{Introduction}

Online social networks such as Facebook and Twitter are now an ubiquitous part of modern life. 
Moreover, given the prevalence of economic situations in which the network of relationships between agents play an important role in outcomes, it is essential to rigorously understand how networks form and what network structures are likely to emerge.  
Large interdisciplinary subfields that combine economics, sociology, mathematics and computer science in the study of social networks are emerging (see \cite{csw2005} for a survey). 
While many models for social network exist, most are either stochastic (i.e., probabilistic models) or are learned models (i.e., constructed by fitting a set of parameters). 
The game theoretic approaches to network formation that exist are  largely motivated by games where network infrastructure is being built and costs are shared amongst agents (see, e.g., {\cite{AGTbook}} Chapter 19), and do not necessarily capture natural properties of online social networks. 
We introduce a simple directed network model that has a natural interpretation with respect to many online social networks. 
In this model the agents are nodes in the network and the model is defined by three key parameters:
\begin{enumerate}
\item the \emph{cost} $c_s$ of directly connecting to another agent (i.e., making a friend request),  
\item the \emph{cost} $c_\ell$ of accepting a connection another agent (i.e., confirming a friend request),  and 
\item the \emph{distance} $k$ (i.e., the maximum path length) that suffices for gaining utility from an indirect connection to another agent.
\end{enumerate}
Agents trade off decisions between the cost of maintaining edges against the rewards (in terms of connectivity) from doing so.
Allowing $c_s, c_\ell > 0$ captures many online social networks such as Facebook and LinkedIn in which one agent initiates a connection request and the other choses to accept or decline. 
When $c_\ell=0$, the model captures other online social networks such as Twitter in which a connection can be made unilaterally. 
The distance $k$ captures the maximum path distance that suffices for deriving utility from (indirect) connections;  a generalization of this model can further allow \emph{target sets} T(v) which defines the set of agents that $v$ would like to reach within distance $k$.\footnote{See Section~\ref{sec:conclusion} for a discussion -- in this paper we only consider $T(v) = V$ for all agents $v$.}

In particular, we study natural dynamics in which agents periodically make asynchronous decisions on whether to add or sever edges. 
We show that irrespective of the initial network, the dynamic process converges; in fact when $c_s, c_\ell > 0$ the convergence is \emph{fast}. 
Furthermore, the fixed points of this process form a nontrivial class of networks.\footnote{This is in stark contrast to related settings in which the only fixedpoints are cycles and empty graphs \cite{BaGo2000}.}
We further study the static game, and prove that for a nontrivial range of parameters a \emph{flower graph} is efficient and stable, and a \emph{Kautz graph} is symmetric-efficient. 
We leave open the technically challenging question of whether symmetric networks that are also stable exist for nontrivial parameter ranges. 
Finally, we study the extremal properties of stable networks and show that the price of anarchy can be arbitrarily bad, although the price of stability is 1.

Lastly, this model allows us to define a multi-dimensional generalized clustering coefficient which relates to the \emph{stability} of the network; this could be of independent interest in the study of social networks 
{\cite{wf1994}}.

\section{Preliminaries}
\label{sec:prelim}


Let $V$ be a set of agents with $n = |V|$. 
\begin{defn}[Bidirected Network]
A bidirected network $G = (V, E_s, E_\ell)$ is a graph with vertex set $V$ with two \emph{types} of directed edges; {speaking} edges $s_{uv} \in E_s$,\footnote{An edge $s_{uv}$ can be thought of as $u$ \emph{initiating} contact with $v$.} and {listening} edges $\ell_{vu} \in E_\ell$.\footnote{A listening edge $\ell_{vu}$ can be thought of as $v$ \emph{accepting} contact from $u$.}
\end{defn}
Note that $s_{uv}$ exists independently of $\ell_{uv}$; both, one, or none may be present in the network. When clear from context, with some abuse of notation, we drop the $s/\ell$ demarkations and simply refer to edges $uv \in E$.
We let $ds^+_G(v) = |\{w | s_{vw} \in E_s \}|$ denote the number of outgoing speaking edges of $v$, and let $ds^-_G(v) =|\{w | s_{wv} \in E_s \}|$ denote the number of incoming speaking edges to $v$. The analogous  definition is used for listening out-degree ($d\ell^+_G(v)$) and in-degree ($d\ell^-_G(v)$). 

\begin{defn}[Speaking and Listening Reachability]
We say there is a \emph{speaking path of length $k$} from $v$ to $u$ if there exists a set of directed edges $s_{vv_1}, s_{v_1v_2}, \ldots, s_{v_ku} \in E_s$ and edges $\ell_{uv_k}, \ell_{v_kv_{k-1}}, \ldots, \ell_{v_1v}\in E_{\ell}$. We say that a vertex $u$ that has a speaking path of length at most $k$ from $v$ is \emph{$k$-speaking-reachable from $v$}, and let $R^s_{G,k}(v) \subseteq V$ be the set of all such vertices. Listening paths, listening reachability, and the set of listening-reachable vertices  $R^\ell_{G,k}(v) \subseteq V$ are defined in an analogous manner.
\end{defn}
With some abuse of notation, when $k$ is clear from context we drop it from the notation above.  Note that if $u$ is speaking-reachable from $v$ then $v$ is listening-reachable from $u$.

\subsection{Our Model}

Each agent $v \in V$ has a strategy $S_v = (S^s_v, S^\ell_v)$, which consists of subsets of agents $S^s_v, S^\ell_v \subseteq V$. Thinking of agents as vertices, $S^s_v$ (respectively $S^\ell_v)$ corresponds to the set of vertices that $v$ connects to by building speaking (respectively listening) edges $s_{vu}$ (respectively $\ell_{vu}$). Thus, the strategy vector $\mathbf{S} = (S_1,\ldots, S_n)$ defines a bidirected graph $G = (V,E_s, E_\ell)$ where $E_s = \{vu | u \in S^s_v\}$ and $E_\ell = \{vu | u \in S^\ell_v\}$. With some abuse of notation, we often refer to $G$ as the set of strategies and will use $G$ and $\mathbf{S}$ interchangeably.

The utility of $v$ is given by
\[U_G(v) = U^s_{G,k}(v) + U^\ell_{G,k}(v) \]
where 
\[ U^s_{G,k}(v)  = \vert R^s_{G,k}(v) \vert - c_s \cdot  ds_G^+(v) \vert  \mbox{ and } \quad 
 U^\ell_{G,k}(v)  = \vert R^\ell_{G,k}(v) \vert - c_\ell \cdot  d\ell_G^+(v)  \]
are the utilities derived from speaking and listening respectively. The \emph{costs} $c_s$ and $c_\ell$ capture the cost of maintaining speaking and listening edges respectively. 

A natural special case is that in which one of the costs is 0 (without loss of generality $c_\ell = 0$).  For such a model, an agent can always set $S^\ell_v = V$ without loss to her utility. Hence, the strategy space boils down to $S^s_v$. Moreover, we can consider only the speaking portion of the utility $U_{G,k}(v) = U^s_{G,k}(v)$ without loss of generality. In such cases we drop all $s/\ell$ demarkations and simply think of a directed graph $G = (V,E) = (V, E_s)$. This special case, when we further assume that $k = \infty$, is equivalent to the network model in \cite{BaGo2000}. 

\subsection{Dynamics}

\begin{figure}[t]
\begin{center}
\includegraphics[width = 2.5in]{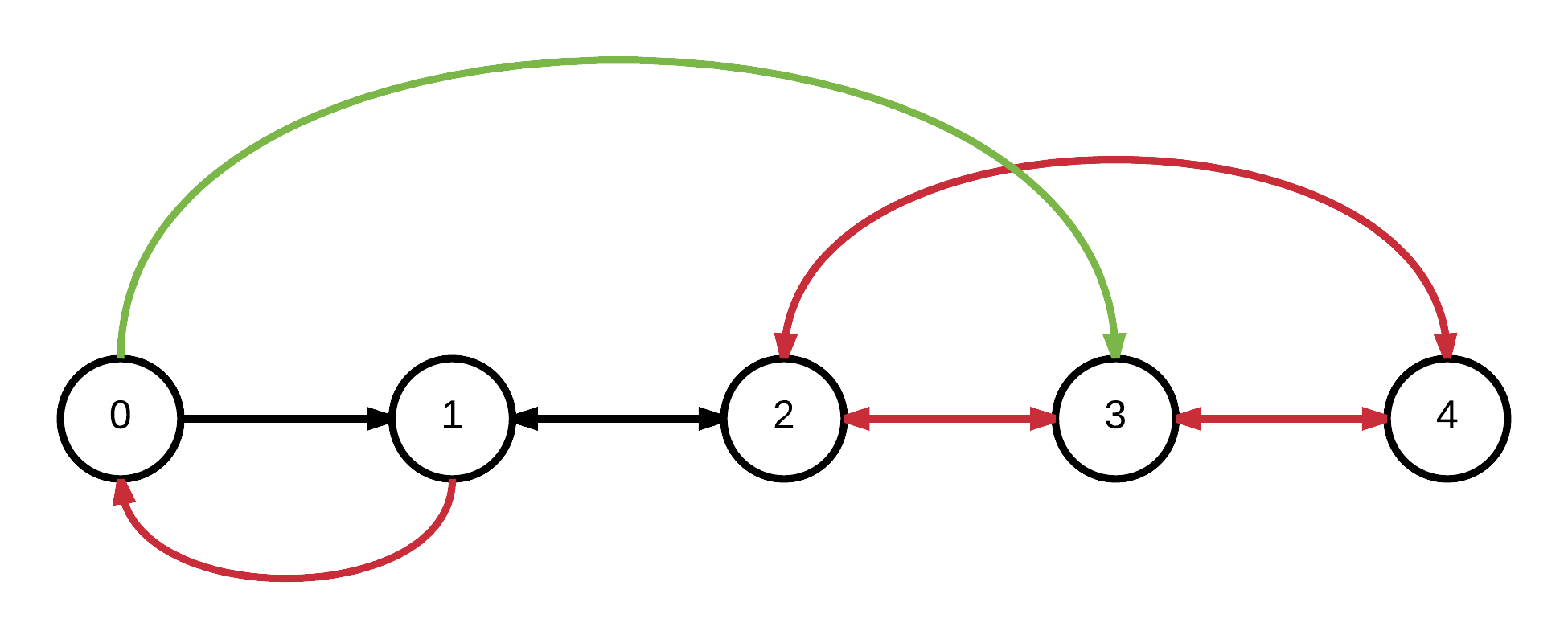}
\caption{Example network depicting addable edges (green), removable edges (red) and existing edges that are not removable (black) for $k=2$, $c = c_s=1.5$ and $c_\ell = 0$; note that we only depict speaking edges (one can assume that all listening edges are present).}
\end{center}
\label{fig:example}
\end{figure}

In this paper we only consider dynamics that are asynchronous (i.e., one agent updates at a time) and stochastic (i.e., agents update in a random order). 
A shorthand notation for the network obtained by adding (alternatively, deleting) the edge $vw$ from an existing network $G$ is $G + vw$ (alternatively, $G-vw)$. Similarly, we let $G + S_v$ be the network obtained by adding all $vu$ edges where $u \in S_v$ to $G$. 
The following definition is convenient for a variety of our definitions and results.
\begin{defn}[Addable and Removable Edges]
We say an edge ${uv} \in G$ is {removable} if $U_{G-{uv}}(u) > U_G(u)$. Similarly, we say an edge $uv \not\in G$ is addable if $U_{G+{uv}}(u) > U_G(u)$.
\end{defn}
See Figure~\ref{fig:example} for an example of a network where the addable and removable edges are depicted. 
We can then define the \emph{edge dynamics} as follows: 
\begin{defn}[Edge Dynamics]
In each round, one potential (speaking or listening) edge $vw$ is selected at random. Without loss of generality, assume it is a speaking edge. If $s_{vw} \in E_s$ then the edge is deleted if and only if it is removable. Alternatively, if $s_{vw} \notin E_s$ then the edge is added if and only if it is addable. The analogous  definition is used for listening edges. 
\end{defn}
The bulk of our convergence results concern edge dynamics. In some cases they also apply to a natural vertex dynamics (see Definition~\ref{def:vertex_dynamics}).

\begin{rem}
Note that we do not consider \emph{best-response dynamics} in which, in each time step, a vertex $v$ is selected at random and she updates update her strategy $S_v$ in order to maximize her utility (potentially by changing multiple edges simultaneously) with respect to the current strategies $\mathbf S_{-v}$ of the other agents. Indeed, our results do not extend naturally to this setting. For some special cases of our model, if the dynamics are performed \emph{synchronously}, the results from \cite{BaGo2000} apply.
\end{rem}

\subsection{Stability and Efficiency}

Denote by $\mathbf{S}_{-v}$ by the $(n-1)$-dimensional vector of the strategies played by all agents other than $v$.  With some abuse of notation we  use $({S_v,\mathbf S_{-v}})$ and $\mathbf S$ and $G$ interchangeably as is   convenient. 
\begin{defn}[Stability]
A strategy vector $\mathbf{S}$ is said to be {stable} if for all agents $v$ and each potential strategy $S'_v \subseteq V$ , we have that 
\[ U_{S_v,\mathbf S_{-v}}(v) \geq U_{S'_v,\mathbf S_{-v}}(v) .\]
This is equivalent to saying that $\mathbf S$ is a Nash equilibrium. 
\end{defn}
In other words, no agent $v$ has any \emph{incentive} to change her strategy from $S_v$ to $S'_v$, assuming that all other agents stick to their current strategies. Observe that such a solution is self-enforcing in the sense that once the agents are playing such a solution, no one has any incentive to deviate.
{In fact, for our model, something stronger holds: 
\begin{prop}
A  strategy vector is stable if and only if no edge is addable or removable.
\end{prop}
\emph{Pairwise stability} is a common strengthening of the notion of stability. It is natural in social networks where, effectively, a link between two agents is formed only if both endpoints are in agreement, but either can unilaterally delete an edge. In our model, an agent's utility is never decreased by an incoming edge, hence there is no difference between stability and pairwise stability. However, in the bidirected case (when $c_s, c_\ell > 0$), an extended notion of pairwise stability where a speaking edge and its corresponding listening edge are consider in conjunction is natural.
\begin{defn}[Bidirected Pairwise Stability]
A strategy vector $\mathbf{S}$ is said to be {bi-pairwise stable} if for all pairs of agents $u, v$ 
\[ U_{\mathbf S - s_{uv} }(u) \leq U_{\mathbf S }(u) , \; U_{\mathbf S - \ell_{uv} }(u) \leq U_{\mathbf S }(u), \]
and if 
\[ U_{\mathbf S + s_{uv} + \ell_{vu}}(u) > U_{\mathbf S }(u) \; \mbox{ then } U_{\mathbf S + s_{uv} + \ell_{vu}}(v) < U_{\mathbf S }(v).
\]
\end{defn}
For the remainder of this paper we refer to this notion as pairwise stability.

Often, notion of {fairness} or {global optimality} are important considerations.  
The utilitarian objective welfare of a set of strategies is the \emph{collective utility} of all of the agents; i.e., for the strategy set $G$, it is $\varphi(G) = \sum_v U_G(v)$. 
\begin{defn}[Efficiency and Symmetry]
We say a set of strategies $\mathbf S$, and the network $G$ it defines, is  {efficient} if it maximizes $\varphi(G)$. It is {symmetric} if $U_G(v) = U_G(u)$ for all $u,v \in V$, and is {asymmetric} otherwise. It is symmetric-efficient if it maximizes $\varphi(G)$ over to the set of symmetric networks. 
\end{defn}
Note that, a priori, efficiency, symmetry, and stability need not be satisfied simultaneously. One aim of this work is to explore these relationships for our model.


\subsection{Preliminary Observations}

Before stating our main results, we make a few preliminary observations which will become useful in later proofs.
\begin{lem}
\label{pairwise_stable_if_stable}
If $G$ is bi-pairwise stable then $G$ is stable.
\end{lem}
\begin{proof}
For sake of contradiction, assume that $G$ is pairwise stable but is not stable.
Since $G$ is pairwise stable so for any speaking edge $s_{vw} \in E_s$ we have $U_G(v) \geq U_{G-s_{vw}}(v)$ and for any listening edge $\ell_{vw} \in E_\ell$ we have $U_G(v) \geq U_{G-\ell_{vw}}(v)$. Hence, $G$ has no removable edge. And because we assume that $G$ is not stable so there must exist a speaking edge $s_{vw} \not\in E_s$ where $U_G(v) < U_{G+s_{vw}}$ or a listening edge $\ell_{vw} \not\in E_\ell$ where $U_G(v) < U_{G+\ell_{vw}}(v)$ which contradicts the definition of pairwise stability.
\end{proof}

\begin{defn}[Complete Edges]
\label{def:complete}
With some abuse of notation, we say a speaking edge $s_{vw}$ is complete if 
\[s_{vw} \in E_s \Rightarrow \ell_{wv} \in E_\ell \] We say a listening edge $\ell_{vw}$ is complete if  
\[\ell_{vw} \in E_\ell \Rightarrow s_{wv} \in E_s .\] 
\end{defn}
 In any stable or efficient network, if $c_s, c_\ell > 0$, all edges are complete. Hence, despite allowing unilateral actions, agreement naturally emerges.
\begin{lem}
\label{prop:bi_stable}
If $c_s>0, c_\ell>0 $, then for any stable, pairwise stable, or efficient network $G$ all speaking and listening edges are complete.
\end{lem}
\begin{proof}
We will prove the contrapositive. Assume that there exist vertices $v, w \in V$ such that $s_{vw} \in E_s $ and $\ell_{wv} \notin E_\ell$. 
Since $\ell_{wv} \notin E_\ell$, if for any node $z$ we have $z \in R^s_{G,k}(v)$, 
then, by the definition of reachability, there exists a path from $v$ to $z$ that does not use edge $s_{vw}$. 
Thus there is no node $z$ such that $z \in  R^s_{G,k}(v)$ and $z \not\in  R^s_{G-s_{vw},k}(v)$. 
In other words $\vert  R^s_{G-s_{vw},k}(v) \vert = \vert  R^s_{G,k}(v) \vert$, and the positive component of the utility is the same for $v$. However, $ds^+_{G-s_{vw}}(v) = ds^+_{G}(v)-1$, so the negative component of the utility is reduced by $c_s$ in $G-s_{vw}$. Thus, $U_{G-s_{vw}}(v)= U_G(v)+c_s$ and, since $c_s>0$, we conclude that $U_{G-s_{vw}}(v) > U_G(v)$. Hence, $G$ is could not be stable or pairwise stable. Moreover, deleting $s_{vw}$ does not affect on the utility of other vertices so $\varphi(G-s_{vw}) > \varphi(G)$. Hence $G$ is not efficient.
The proof follows analogously for $\ell_{vw}$.
\end{proof}

%

\section{Related Work}



Due to the vast range of applications; from sociology to commerce, biology and physics, with drastically different underlying properties, many models have been developed and studied in depth (see  \cite{Newm2003} for a survey).
Starting with $G(n,p)$ \cite{ErRe1960,Boll2001}, stochastic models have often taken a forefront. 
Depending on the observed graph properties, different models take the forefront, such as preferential attachment models \cite{BA1999} for specific degree distributions, or small-worlds models \cite{WaSt1998} for capturing social networks.  An alternate approach, is to take an existing network and fit a model using techniques from machine learning. For example, a the authors of \cite{KJJFYY2008} attempt to understand the Twitter network by fitting a stochastic model.
However, while stochastic  and learned models can explain on a macro level what is occurring in a network, on a micro level, i.e., looking at individual nodes and its edges, they remain uninformative; the motivation as to why a node would maintain an edge is abstracted away. We instead consider game theoretic models of a network in which each node is a selfish agent and decided if and whom to connect to based on her utility.


There has been a lot of very interesting work on network formation (see \cite{Jack2005} and {\cite{AGTbook}} Chapter 19 for nice surveys coming out of the Economics and Algorithmic Game Theory literature respectively). Myerson was the first to consider such models (see, e.g., \cite{Myer1977,AuMy1988}). However, formulated the problem as a \emph{cooperative} game where agents worked together towards a common goal; in our setting we assume agents have individual, or \emph{selfish}, goals. In another line of work, global connection games (see \cite{AGTbook}) are studied in which agents are not nodes in the network, rather vested parties in (individual) global connectivity properties of the game. 

More closely related models consider selfish agents that are nodes in \emph{undirected} networks. 
\cite{JaWo1996} introduced a model for the study of the (static) stability of undirected networks. Also known as the \emph{local connection game} (see \cite{AGTbook}), nodes have discounted (based on path length) rewards for being connected to another agent, and cost for making a link. Their goal was to understand the relationship between stability and efficiency, which led to further results in this direction (see, e.g., \cite{DuMu1997,Jack2001}). 
The authors of \cite{Watt1999} consider edge dynamics for this undirected model. This is further studied by \cite{JaWa2002} where it is shown that the stochastic best response dynamics may not converge; this is in contrast to our model which will always converge to a stable network.

 Directed networks allow one individual to connect
to another without the consent of the second individual, and thus applications are to settings such as Twitter \emph{following}, while undirected network capture social networks where links are reciprocal, such as Facebook \emph{friendship}. 
%
The difference between directed and undirected graphs is not just a technicality when it comes to modeling.  
In undirected networks, edges are implicitly \emph{reciprocal}, hence consent is required from both endpoints; thus, undirected models are suitable for many forms of economic and social relationships. For directed networks, however, a vertex can link directly to another without reciprocally; thus, directed models are more suitable for capturing interactions that are passive in one direction, as with the consumption of public content.  
Not only are the applications and modeling considerations distinct, but they can lead to dramatic differences in the resulting graphs, their properties, and their efficiency.

The closest related work to our setting, by Bala and Goyal \cite{BaGo2000},  studies the stability, efficiency, and dynamics of \emph{directed} networks. In fact, their can be viewed as a special case of ours when $c_\ell = 0$ and $k = \infty$. However, their dynamics differ significantly from ours; they use lazy simultaneous best-response dynamics while we consider  asynchronous stochastic edge dynamics. Due to the nature of their update, their process always converges to either a cycle or the empty network. Not only are our dynamics more natural because coordination is not required and connections are evaluated on an individual basis, but the class of networks that are fixed points of our dynamics is nontrivial.


\section{Dynamics}

In studying the dynamics, different approaches are required for the special case where $c_\ell = 0$ (i.e., the graph is directed) and more general bidirected setting (where $c_s, c_\ell > 0$). In the latter, an agent's hand is often forced; as reachability can only occur via paths of complete edges, $u$ never has incentive to add an edge $s_{uv}$ if $\ell_{vu}$ is not present (a layman's interpretation is to say that one cannot accept a connection that is not initiated). This, in effect, speeds convergence and makes our work easier. 
We consider both settings in Sections~\ref{sec:dir_dyn} and \ref{sec:bidir_dyn} respectively.

\subsection{The Directed Setting ($c_\ell = 0$)}
\label{sec:dir_dyn}

In this section we consider the setting where $c_\ell = 0$. Recall that in this case we can assume $S_\ell = V$ and the model reduces to that for a directed graph with only speaking edges. Moreover, for the directed setting, our results only hold for $k = \infty$. Hence, for ease of notation we drop all $s$, $\ell$ and $k$ demarkations, e.g., $G = (V, E) = (V, E_s)$, $c = c_s$, $U_G = U_{G,k}^s$, and $uv = s_{uv}$.

\begin{thm}
\label{thm:converge}
The edge dynamics for $k = \infty$ converge to a stable graph. 
\end{thm}
It is a priori not clear that any dynamics should converge, or that any interesting structure could emerge at a fixed point of the dynamics.\footnote{E.g., recall that  the authors of \cite{BaGo2000} find that synchronous best-response dynamics must converge to either the cycle or the empty graph under the same conditions.}
On the other hand, as is evident in our proof, a nontrivial class of networks are fixedpoints of the edge dynamics. 
The entire proof of convergence is rather technical and requires an in-depth case analysis which we have split out into individual lemmas for the reader's clarity.
\begin{proof}[of Theorem \ref{thm:converge}]
In order to prove convergence, it suffices to show that for any initial graph $G_1$, there exists a finite sequence of graphs $G_1, \ldots, G_s$ such that $G_i \neq G_j$ for any $i \neq j$ and $G_i$ and $G_{i+1}$ differ by a single edge that is either addable or removable in $G_i$ such that $G_s$ is stable. In other words, there is a \emph{path to convergence} from any state which only uses addable or removable edges. If this path is finite, there is some (albeit tiny) probability, which is bounded away from 0, that this sequence of edges is selected. Hence, it suffices to construct such a sequence.
We  construct it as follows:
\begin{enumerate}[topsep=1pt,itemsep=0ex,partopsep=0ex,parsep=1pt]
\item Recursively delete removable edges one at a time until no removable edge remains. 
\item If there is no addable edge then we have reached a stable graph. 
\item Otherwise, partition the graph into maximal strongly connected components; the component-level graph must be a directed acyclic graph (Lemma~\ref{lem:dag}). We call a component a \emph{root} if it has no incoming edge, and a \emph{leaf} if it has no outgoing edge. It is often convenient for us to restrict to the components of size $> c$, which we call large-components; roots and leaves are then defined within the large-component graph. Note that there must be at least one large component, otherwise we contradict the fact that an edge is addable (Lemma~\ref{lem:onebig}).
\item For each root $T_i$ of the large-component graph, designate a special vertex $r_i$. Note that for any (large or small) component $C$ such that $T_i \not\subseteq R_G(C)$, the edge $xr_i$ is addable for any $x \in C$ (Lemma~\ref{lem:addtoroot}); hence, the edges specified in steps 5-7 must be addable.
\item If some leaf $L_i \neq T_i$ exists in the large-component graph rooted at $T_i$, add the edge $\ell_i r_i$ for some $\ell_i \in L_i$ and go back to step 1. 
\item Otherwise, all large-components have no edges between them in the large-component graph. If there are more than one large components, let $T_1, T_2$ be any two large-components such that $T_1 \not\subseteq R_G(T_2)$. Add the edge $r_2r_1$. If $T_2 \not\subseteq R_{G+r_2r_1}(T_1)$, also add the edge $r_1r_2$. Go back to step 1. 
\item Otherwise, there is exactly one large-component $T_1$. If there exists a small component $S_j$ that is a leaf, add the edge $s_jr_1$ for some $s_j \in S_j$ and go back to step 1.
\item Again, there is exactly one large-component $T_1$, and this set $T_1$ is reachable from every component (otherwise we would be in step 7). Moreover, there must be at least one small component $S_k$ that is a root and in which $|R_G(S_k) \setminus T_1| > c$ (Lemma~\ref{lem:step8}). Let $t_kr_k$ be an edge that reaches $T_1$ on a path from $S_k$, i.e.,  $r_k \in T_1$ and $t_k \not\in T_1$, and $t_k \in R_G(S_k)$. Add the edge $r_ks_k$ for some $s_k \in S_k$ and go back to step 1.
\end{enumerate}

We prove in Lemma~\ref{lem:addinsamecomp} that the only addable edges are between two different components; hence only inter-component edges must be addable; the above steps cover all possible types of addable edges in sequence until none remain.

After step 5+1 (alternatively 6+1), the number of large components that existed before step 5 (alternatively 6) is reduced by at least one (see Lemma~\ref{lem:rem5} and Lemma~\ref{lem:rem6} respectively). Hence, we will reach step 7 in finitely many rounds.

In step 7 there is exactly one large component (see also Lemma~\ref{lem:onebig}), and the number components \emph{without} a direct edge to $T_1$ is reduced by one. Moreover, after step 7 there are no removable edges, so the number of large components and the size of the largest component does not change in step 7+1 (Lemma~\ref{lem:rem7}). Thus, we never go back to steps 5 or 6, and after finitely many rounds, we will move on to step 8. 

In step 8 again there is exactly one large component. In step 8+1 the the number of large components does not increase, thus, we never go back to step 5 or 6. While we may go back to step 7, however, every time we complete step 8 we have removed at least one edge from being addable at any point in the future (Lemma~\ref{lem:rem8}). Hence, we can remain in steps 7 and 8 for finitely many rounds, and the process will terminate in step 2. 
\end{proof}

\begin{lem}
\label{lem:rem8}
After step $8$, we removed at least one edge from being addable at any point in the future. Moreover, the number of large components does not increase. 
\end{lem}
\begin{proof}
By the same argument as in the proof of Lemma~\ref{lem:rem5}, we know that step $8+1$ cannot increase the number of large components. 

We prove that the edge $r_ks_k$ cannot be added again in the future.
Let $T_1$ be the largest component in $G$ (before step 8) and let $T_1''$ be the largest component in $G''$ (after step 8+1). We reserve $T_1'$ and $G'$ for an intermediary graphs. 

It now suffices to argue that $T_1'' \in R_{G''}(s_k)$; hence either $s_k \in T_1''$, or $r_ks_k$ was removed in step 1 (after we added it in step 8), in which case $|R_{G''}(s_k) \setminus T_1''| < c$. Since steps 7 and 8 cannot increase $R_{G''}(s_k)$ for any future $G''$, the edge $r_ks_k$ will never become addable again. 

For sake of contradiction, assume $T_1'' \not\in R_{G''}(s_k)$. Let $xy$ be the first removable edge such that  $T_1' \not\subseteq R_{G'-xy}(s_k)$; hence, $xy$ must have been on every path from $s_k$ to $T_1'$, and in particular $x \in R_{G'}(s_k)$ and $T_1 \not\subseteq R_{G'-xy}(x)$. However, $|T_1| > c$, hence $xy$ is not removable which gives a contradiction.  
\end{proof}

\begin{lem}
\label{lem:addinsamecomp}
If $vw$ is an addable edge in $G$, then $v$ and $w$ belong to different strongly connected components of $G$.
\end{lem}
\begin{proof}
An edge $vw$ is addable if $U_{G+vw}(v) > U_G(v)$, in other words $|R_{G+vw}(v) - R_G(v)| \geq c$. Assume that $v$ and $w$ belong to same strongly connected component. Thus $R_G(v) = R_G(w)$, and hence $R_{G+vw}(v) = R_G(v)$. Therefore the positive component of the utility for $v$ does not increase, while the negative component would decrease. Hence no such $vw$ edge is addable. 
\end{proof}

\begin{lem}
\label{lem:dag}
In step 3, the component-level graph is a directed acyclic graph. 
\end{lem}
\begin{proof}
Assume for sake of contradiction, that some set of two or more components form a cycle in the component-level graph. Then, for every $x, y$ in the cycle, there must be some path from $x$ to $y$ as there is a path between any two vertices within every component, and a path between any two components via the cycle. Hence, the cycle in fact forms a strongly connected component contradicting the maximality of the partition. 
\end{proof}

\begin{lem}
\label{lem:onebig}
Any leaf must either be an isolated vertex (with no parent) or is of size at least $c$; in particular a leaf that is not a root must be of size at least $c$. Moreover, in step 3, there exists a component of size at least $c$. 
\end{lem}
\begin{proof}
If there is an edge $ab$ such that $|R_G(b)| < c$, then by removing $ab$, the positive component of the utility decreases by less than $c$ while the negative component decreases by $c$. Hence, $ab$ is removable, and would have been eliminated in step 1. Thus, any edge $ab$ must have $|R_G(b)| \geq c$. Consider a component $B$, and consider any edge $ab$ such that $b \in B$. If some such edge exists, then $R_G(b) = R_G(B)$, and hence $|R_G(B)| \geq c$ as shown above. If no such edge exists, then $B$ must be an isolated vertex.

In step 3, there exists some edge $ab$ that is addable. From Lemma~\ref{lem:addinsamecomp}, we know that $a$ and $b$ must be in different components $A$, $B$. By assumption, we know that $|B| < c$, however, $|R_G(b) \setminus R_G(a)| \geq c$. Thus, there must be at least one other component that is reachable from $B$. Thus, the component graph has at least one leaf, and, from above, that leaf must be of size at least $c$. Thus, there exists at least one component of size at least $c$.
\end{proof}

\begin{lem}
\label{lem:addtoroot}
In step 4, given any component $C$ such that $T_i \not\subseteq R_G(C)$, for any $x \in C$ the edge $xr_i$ is addable.
\end{lem}
\begin{proof}
Let $T_i$ be a root of the large-component graph; trivially $|T_i| \geq c$. If $T_i$ is not reachable from a component $C$, then, for any $x \in C$ we know $|R_{G+xr_i}(x) \setminus R_{G}(x)| \geq |R_{G}(r_i) \setminus R_{G}(x)| \geq |T_i| \geq c$. Thus, $xr_i$ is addable.
\end{proof}

\begin{lem}
\label{lem:rem5}
After adding an edge as in step $5$ and then completing step $1$, the number of large components at the beginning of step $5$ has reduced by at least 1. 
\end{lem}
\begin{proof}
Clearly, adding edges cannot increase the number of components. Moreover, since $L_i$ is a leaf of $T_i$, and $L_i$ and $T_i$ there is clearly a path from every $r_i \in T_i$ to every $x_i \in L_i$ in the original graph. Using the new edge $\ell_ir_i$, there is now also a path from every $x_i \in L_i$ to every $r_i \in T_i$. Hence this reduces the number of large components by at least one.

Let $G'$ be the current graph (after adding $\ell_ir_i$, and removing an unknown number of edges). We must now show that, after deleting an edge, the number of large components has not increased. 

Let $xy$ be a removable edge in $G'$ and let $X$ be the component that contains $x$. If $y \not\in X$, then removing $xy$ does not increase the number of components of $G'$. Assume $y \in X$. Let $S$ be the set of vertices that are reachable from $x$ only through $xy$. Thus, $S \cap X$ and $X \setminus S$ will form two new strongly connected components after the removal of $xy$. 
Because $xy$ is removable, we know that $|R_{G'-xy}(x) \setminus R_{G'}(x)| < c$. Note that we have exactly defined $S = R_{G'-xy}(x) \setminus R_{G'}(x)$. Hence, $|S \cap X| \leq |S| < c$. Thus we have not increased the number of large connected components in $G'$.

Thus, overall, the number of large components reduces by at least one. 
\end{proof}

\begin{lem}
\label{lem:rem6}
After adding an edge as in step $6$ and then completing step $1$, the number of large components at the beginning of step $6$ reduces by at least one. 
\end{lem}
\begin{proof}
Clearly, adding edges cannot increase the number of components. Moreover, after adding $r_2r_1$ (and if necessary $r_1r_2$), the two large components $T_1, T_2$ are in the same large component in the new graph. Hence, the number of large components has reduced by at least one. 

By the same argument as in the proof of Lemma~\ref{lem:rem5}, we see that after deleting edges, the number of large components does not increase, and hence, overall, the number of large components has reduced by at least one.
\end{proof}

\begin{lem}
\label{lem:rem7}
After adding an edge as in step $7$ and then completing step $1$, the number of large components at the beginning of step $7$ does not increase and the number of components without a direct edge to $T_1$ decreases by one.
\end{lem}
\begin{proof}
Note that, at the beginning of step 7 there is a single large component (otherwise we would be in step 5 or 6).
Clearly, adding  the edge reduces the number of components without a direct edge to $T_1$, and adding edges cannot increase the number of components or create more roots. We now show that no edge is removable when we go to step 1, everything that holds after step 7 must also hold after step 7+1. 

Let $G$ be the graph at the beginning of step 7. For sake of contradiction, assume that $xy$ is a removable edge; thus, $|R_{G+s_jr_1-xy}(x) \setminus R_{G+s_jr_1}(x)| < c$. Since $s_jr_1$ was addable in $G$, we know that $xy \neq s_jr_1$. Additionally, since $s_j$ was an isolated vertex in $G$ we know that it has no other edge and hence $x \neq s_j$. Moreover, since $s_j$ was an isolated vertex in $G$, we know it has no incoming edge, and hence $x \not\in R_{G+sjr_1}(z)$ for any $z \neq x$. Thus, $R_{G+s_jr_1-xy}(x) = R_{G-xy}(x)$ and $R_{G+s_jr_1}(x) = R_{G}(x)$, and this implies that $|R_{G-xy}(x) \setminus R_{G}(x)| < c$, and hence $xy$ was removable in $G$. This gives a contradiction, as no such $xy$ exists at the beginning of step 7 (as it would have been removed in the previous step 1).
\end{proof}

\begin{lem}
\label{lem:step8}
There  is exactly one large-component $T_1$, this set $T_1$ is reachable from every component, and there must be at least one small component $S_k$ with no incoming edge such that $|R_G(S_k) \setminus T_1| > c$.
\end{lem}
\begin{proof}
Exactly one large-component $T_1$, since at least one large component exists (Lemma~\ref{lem:onebig}) and if two or more exists we  would be in step 5 or 6.
Moreover, $T_1$ is reachable from every component, otherwise we would be in step 7.
Lastly, we know that there must be some small component $S_k$ such that $|R_G(S_k) \setminus T_1| > c$, otherwise no edge addable. In particular, any $rs_k$ edge with $r \in T_1$ and $s_k \in S_k$ is addable.
\end{proof}

One important question that remains open is the time until convergence; in particular, we conjecture that the convergence time is \emph{fast}. In effect, this is equivalent to showing that there are many short paths to convergence. Proving this, however, remains a challenging technical open problem.

\subsection{The Bidirected Setting $(c_s, c_\ell > 0)$}
\label{sec:bidir_dyn}

We now go back to the bidirected setting. As observed above, the convergence is less surprising; as edges that are not complete are easily deleted. This is formalized in the following lemma, which will be useful for our result:
\begin{lem}
\label{bidirectional removed edges never added}
If $c_s, c_\ell >0$ and there exist vertices $ v,w\in V$ such that $s_{vw}\notin E_s ,\ \ell_{wv}\notin E_\ell$, then the edge and vertex dynamics will not add either edge back.
\end{lem}
\begin{proof}
We will prove the contrapositive. Assume that some such $s_{vw}$ is added back. We know that $\ell_{wv} \notin E_\ell$. Thus, by the definition of reachability, there is no path from $v$ to any arbitrary node $z$ which passes through $s_{vw}$. Hence, $z \in R^s_{G,k}(v)$ if and only if $z \in S_{G+s_{vw}}^k(v)$. Therefore, $\vert S_G^k(v) \vert = \vert S_{G+s_{vw}}^k(v)\vert $. However, $ds_{G+s_{vw}}^+(v)=ds_G^+(v)+1$ so $U_{G+s_{vw}}(v)=U_G(v)-c_s$ and since $c_s>0$ thus $U_{G+s_{vw}}(v)<U_G(v)$. Therefore agent $v$ does not gain by adding $s_{vw}$ and wo not add it for ever. The similar proof used for listening edges. 
\end{proof}

Using this lemma, we give a very different proof than the one for the directed setting. In addition to working for arbitrary $k$, this proof also allows us to conclude that the convergence is \emph{fast}.

\begin{thm}
If $c_s, c_\ell >0$, then the edge dynamics, vertex dynamics  (see Definition~\ref{def:vert_dyn}) converge to a stable graph in (expected) polynomial time in the number of nodes.
\end{thm}

\begin{proof}
We will prove the theorem for edge dynamics; the analogous proof for vertex dynamics follows from Lemma~\ref{certainty of removing single edge}.
Let $p_G$ be the number of complete edges (see Definition~\ref{def:complete}) and let $m_G$ be the number of edges in network $G$.
We will prove the theorem by induction on $p_G+m_G$.

\paragraph{Base Case:} Suppose that $p_G+m_G=0$, thus $p_G=0,\ m_G=0$ so $G$ is empty and by Lemma \ref{prop:empty}, the empty graph is stable.

\paragraph{Inductive step:}
For sake of contradiction, assume that the dynamics do not converge to a stable graph. Thus, since the number of possible networks is finite, there must exist a sequence of networks that we cycle over. If there exists $v,w \in V$ such that $s_{vw}\in E_s$ and $\ell_{wv}\notin E_\ell$, then with  probability at least $\dfrac{1}{2}\cdot\dfrac{1}{n}\cdot\dfrac{1}{n-1} $, the edge $s_{vw}$ is chosen. By Lemma~\ref{certainty of removing single edge}, this edge is removable for $v$. Therefore, the number of edges is reduced, so $p_{G-s_{vw}}+m_{G-s_{vw}}=p_G+m_G-1$. Thus by the induction hypothesis, $G-s_{vw}$ converges to a stable network. Similarly, if $s_{vw}\notin E_s$ and $\ell_{wv}\in E_\ell$ an analogous proof follows.

Otherwise, for all ${v,w \in V}$, if $\ s_{vw}\in E_s$ then $\ell_{wv}\in E_\ell$ and if $\ell_{vw}\in E_\ell$ then $s_{wv}\in E_s$. Since $G$ is not stable, there must exist an edge which addable or removable. By Lemma~\ref{bidirectional removed edges never added} we know that no removed edge will be added again. Hence, the only possible change is to remove an edge. Suppose that $s_{vw}$ is an existing removable edge for $v$. The probability of selecting $s_{vw} \text{ and } \ell_{wv}$ consecutively is at least $\left(\dfrac{1}{2}\cdot\dfrac{1}{n}\cdot\dfrac{1}{n-1}\right)^2$. Moreover, we know that if $s_{vw}$ and $\ell_{wv}$ are selected in sequence then both will be removed as $s_{vw}$ is removable, and then $\ell_{wv}$ will be removable by Lemma \ref{certainty of removing single edge}. Thus, for new graph we have that $p_{G-s_{vw}-\ell_{vw}}+m_{G-s_{vw}-\ell_{vw}}=p_G+m_G-3$. 

The convergence happens in expected polynomial time because the probabilities, as computed above, shows that after $O(n^4)$ moves in expectation, the value of $p_G+m_G$ will decrease. We know that $p_G+m_G = O(n^2)$. Thus, after at most $O(n^6)$ moves in expectation we will reach a stable network.  
\end{proof}

\begin{defn}[Vertex Dynamics]
\label{def:vertex_dynamics}
\label{def:vert_dyn}
In each round, a vertex $v$ is selected at random, and two options are selected at random 1) speaking or listening and 2) adding or deleting. Without loss of generality, assume speaking and deleting was selected. Then, all removable $s_{vw} \in E_s$ are deleted. Similarly, if speaking and adding is selected, then all addable $s_{vw} \notin E_s$ are added. The analogous  definitions are used for listening edges. 
\end{defn}

\begin{lem}
\label{certainty of removing single edge}
Suppose that $c_s>0$. If there exists a vertex $w \in V$ such that $s_{vw} \in E_s ,\; \ell_{wv}\notin E_\ell$, then if $v$ is selected to remove a speaking edge, she will select some such $s_{vw}$ to remove. The analogous statement holds for listening edges.
\end{lem}

\begin{proof}
When agent $v$ decides to delete a speaking edge, it chooses an edge which has the least gain for it since by removing it, the lost be minimal. We know that the profit of existence a speaking edge is made by providing speaking reachability. Moreover since $\ell_{wv} \notin E_\ell$ and by the definition of reachability, there is no path from $v$ to arbitrary node $z$ that passes through $s_{vw}$ so the gain of it zero. Therefore removing $s_{vw}$ has the least lost for $v$ and is the best strategy among removing all possible speaking edges.
\end{proof}


\section{Stability and Efficiency}
\label{sec:stab}


\begin{figure}[t]
\begin{center}
\hspace{-1in}
\includegraphics[width = 2.36in]{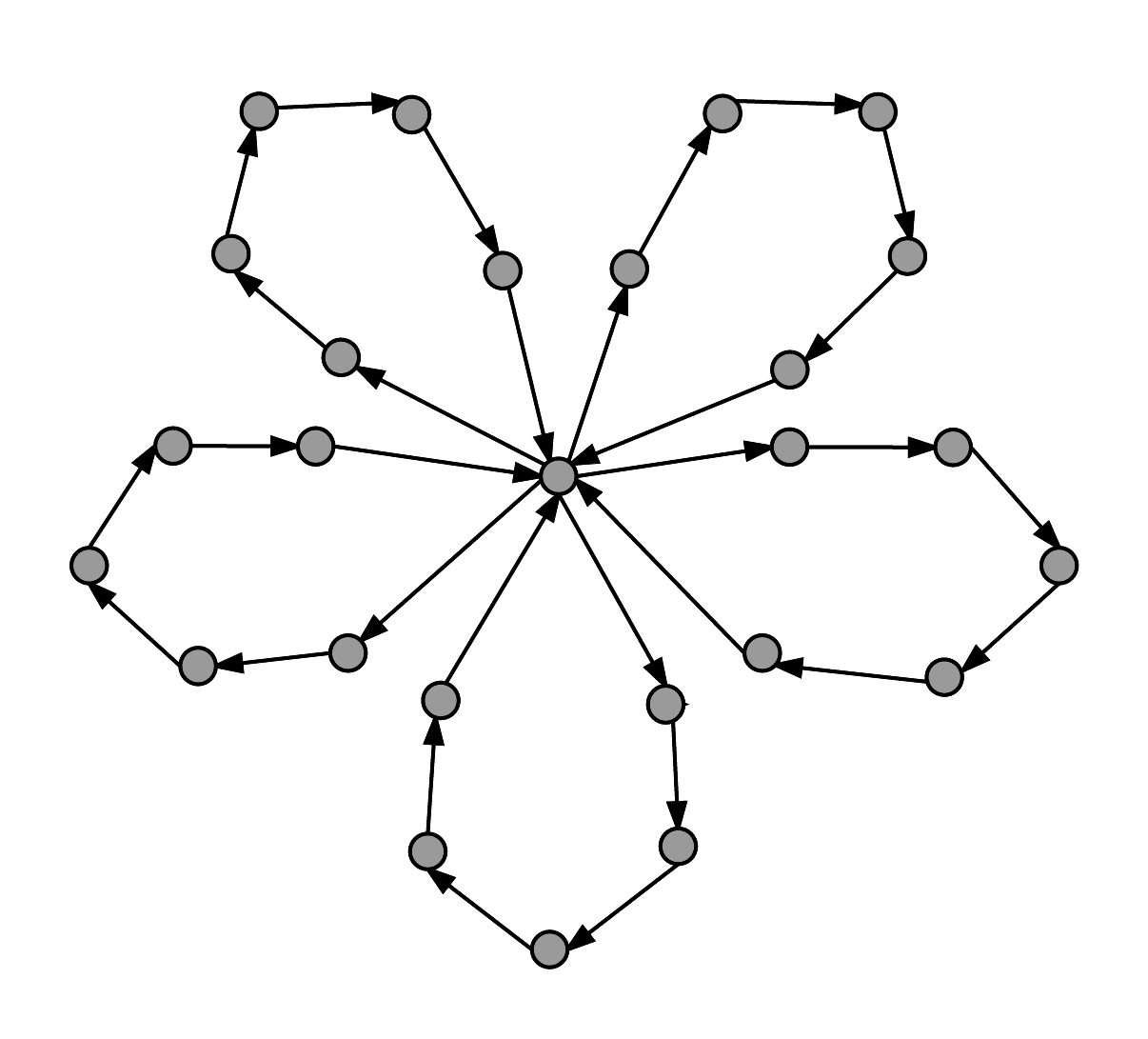}
\hspace{-1in}
\vspace{-.06in}
\caption{A balanced flower graph with $n = 26$ vertices and $k=10$. This is efficient and stable for $c < 5$.}
\label{fig:flower}
\end{center}
\end{figure}

In this section we construct classes of efficient and stable networks. 
We first consider nontrivial ranges of parameters $k > 1$ and  $c_s, c_\ell \leq n-1$ (see Section~\ref{sec:extremal} for a discussion of extremal parameter ranges and networks). 
For this range of parameters, we see interesting networks emerge in the study of efficient and stable networks. 

\begin{rem}
For ease of notation, in the first part of this section we consider only the directed version of our model (i.e., without loss of generality, $c_\ell = 0$). However, all results and proofs follow immediately for the bidirectional case by replacing the directed graphs in question by the analogous bidirected graph where each directed edge is replaced by a complete bidirected edge. In the statement of the theorems, constraints on $c$ apply to both $c_s$ and $c_\ell$.
\end{rem}

The first network we consider, called \emph{balanced flower graph}, is defined for $k \leq 2\sqrt n$ and is constructed as follows: Make a directed cycle of length {$\floor {\xfrac k 2} + 1$} . Select one node from this set and call it the \emph{center}. As long as at least $\floor {\xfrac k 2}$ nodes remain, select them, and, along with the center node, form another directed cycle. Repeat until fewer than $\floor {\xfrac k 2}$ nodes remain; then remove one non-central node from each petal (severing its edges and connecting its predecessor and successor) until you have $\floor {\xfrac k 2} - 1$ nodes and form them into the final petal. We denote by $q$ the number of petals. Note at most one node is removed from each petal in balancing since since $k \leq 2\sqrt n$, and hence {$q \geq \floor {\xfrac k 2}$}.  Note that the balanced flower graph has diameter at most $k$. 
See Figure~\ref{fig:flower} for an example.

\begin{thm}
\label{thm:flower}
For any $4 \leq k \leq 2\sqrt n$ and $3 \leq n$, the social welfare of the balanced flower graph (see Figure~\ref{fig:flower}) is  
{\[ n(n-1)- c\left\lceil{\frac{n-1}{\floor {\frac k 2}}}\right\rceil -c(n-1)\]}
Moreover, if $1 \leq c < \floor{\xfrac k 2}-1$ the balanced flower graph is efficient and pairwise stable.
\end{thm}
Before we prove the theorem, in order to show pairwise stability, a useful lemma is stated:
\begin{lem}
\label{pairwise_stable_if_stable_and_scc}
If $G$ be stable and strongly connected, then $G$ is pairwise stable.
\end{lem}
\begin{proof}
Since $G$ is stable, it contains no removable edge. As $G$ is strongly connected, for any vertex  $v$, the positive component of $U_G(v)$ is $2(n-1)$ and cannot be increased. Hence, there is no addable pair in $G$. Thus, $G$ is pairwise stable.
\end{proof}
\begin{proof}[of Theorem \ref{thm:flower}]
The social welfare of the balanced flower graph is 
{\[ n(n-1)- c\left\lceil{\frac{n-1}{\floor {\frac k 2}}}\right\rceil -c(n-1)\] }
as the utility of the center node is {$(n-1) - cq = (n-1) - c\left\lceil{\frac{n-1}{\floor {\frac k 2}}}\right\rceil$}, and the utility of all other nodes is $(n-1) - c$. {Hence, the social welfare is \[(n-1)\cdot ((n-1) - c) + (n-1) - c\left\lceil{\frac{n-1}{\floor {\frac k 2}}}\right\rceil = n(n-1)- c\left\lceil{\frac{n-1}{\floor {\frac k 2}}}\right\rceil -c(n-1)\]}

Let us now prove that the flower graph $G$ is stable. By Lemma \ref{pairwise_stable_if_stable_and_scc}, this will give pairwise stability. We must show that for all agents $v$, and alternate strategy $s'_v \subseteq V$, we have 
\[ U_{S_v,\mathbf S_{-v}}(v) \geq U_{S'_v,\mathbf S_{-v}}(v).\]

First, consider a vertex $v$ that is not the center; since all vertices are reachable from $v$ by paths of length at most $k$, the positive component of the utility cannot be increased. Hence, utility can only be improved by lowering the cost; however, $|S_v| = 1$, and if $|S'_v| = \emptyset$, then $v's$ utility is 0, which is strictly less than $U_G(v)$. Thus, $v$ does not have an alternate strategy that can improve her utility.

Now consider the center vertex $u$. Again, all vertices are reachable from $u$ by paths of length at most $k$, so the positive component of the utility cannot be increased, and the negative component can only be decreased by decreasing the number of outgoing edges from $u$.
For sake of contradiction, assume there is some $|S'_u| < q$ that is $u$'s best response. Note that $S'_u$ must disconnect $u$ from at least one petal. Since each petal has at least $\floor{\xfrac k 2}-1$ non-center nodes and $c < \floor{\xfrac k 2}-1$, adding the edge from $u$ to the petal to $S'_u$ would increase the utility; this contradicts the fact that $|S'_u| < q$ can be a best response. Hence, the balanced flower graph is stable.

Now we prove that the balanced flower graph is efficient. First, consider the graph that would arise if we simply connected the remainder nodes into a small petal without balancing. This is what is known as simply a \emph{flower graph}. It is known that, for any $n$ and $2 \leq k \leq n$, the flower graph attains the fewest number of edges of any connected graph on $n$ nodes with diameter $k$ (see Theorems 1 and 2 in \cite{HiWo1979}). Note that our balanced flower graph has the same number of edges as the flower graph; hence it also attains the fewest number of edges of any connected graph on $n$ nodes with diameter $k$. Therefore, the balanced flower graph is optimal amongst the set of strongly connected graphs. 

Hence, it suffices to show that a graph that is not strongly connected cannot be efficient. We prove the contrapositive. Assume there is a graph that is not strongly connected; we show that we can combine any two maximal strongly connected components into a single strongly connected component that is at least as efficient. Take two  strongly connected of size $a \geq 1$ and $n-a \geq 1$. Clearly, from above, the social welfare of each component can only improve by making it a balanced flower, and this does not affect the social welfare of the remaining graph. Since it is possible that one component is connected to the other (but not vice versa), without loss of generality the social welfare of this graph is at most  
\begin{eqnarray*}
  a(a-1) &-& c\left\lceil{\frac{a-1}{\floor {\xfrac k 2}}}\right\rceil - c(a-1) 
 + \; (n-a)(n-1) - c\left\lceil{\frac{n-a-1}{\floor {\xfrac k 2}}}\right\rceil - c(n-a-1)  \\
&&\leq n(n-1)-a(n-a) - c\left(\left\lceil{\frac{n-1}{\floor {\xfrac k 2}}}\right\rceil - 2\right)   - cn\\
&&= n(n-1) - c\left\lceil{\frac{n-1}{\floor {\xfrac k 2}}}\right\rceil -c(n-1) + c - a(n-a) \\
&&\leq n(n-1) - c\left\lceil{\frac{n-1}{\floor {\xfrac k 2}}}\right\rceil -c(n-1) + \sqrt n - (n-1),
\end{eqnarray*}
where the last inequality follows as $c < \floor{\xfrac k 2} - 1 \leq \floor{\sqrt n} - 1 \leq \sqrt n$ and $1 \leq a \leq n-1$. For $n \geq 3$, this is less than the social welfare attained by the balanced flower. 
\end{proof}

\begin{rem}
In fact, for any $2 \leq k$, the (unbalanced) flower graph which leaves the last petal at it's size without rebalancing is also efficient; the constraint is due to the fact that for $k > 2\sqrt n$ balancing as described above may not be possible. We could, instead, define a recursive balancing process that continues to steal vertices (possibly several from the same petal) until a balanced flower is reached for some $k' < k$; efficiency follows directly, as does stability for a reduced $c$ which is a function of $n$ and $k$.
\end{rem}

While the above graphs are efficient, they are highly asymmetric with a single node taking on most of the cost. Hence, we now turn our attention towards \emph{symmetric} graphs, and consider a second class of graphs, known as Kautz graphs \cite{BKT1968,II1983} (see Figure~\ref{fig:kautz}).  

The Kautz graph $K_d^{D}$ is a directed graph with $(d +1)d^{D-1}$ vertices. The vertices are labeled by strings $x_0 \ldots x_{D-1}$ of length $D$ with $x_i \in \{0, \ldots, d\}$  with $x_i \neq x_{i+ 1}$.
The set of edges is defined by
\[ \{(x_0, x_1 \ldots x_D,   x_1 \ldots x_D x_{D + 1}) | \; x_i \in \{0, \ldots, d\} \mbox{ and } x_i \neq x_{i  + 1} \}. \]
Clearly, the graph has outdegree $d$, $(d + 1)d^{D}$ edges, and diameter $D+1$.

Kautz graphs arose in the study of the following question: \emph{Given a graph with $n$ nodes and $m$ edges, what is the smallest possible diameter $k$?} Through a series of works, it was shown that Kautz graphs are asymptotically optimal with respect to this question (see \cite{MS2005} for a survey). 

In our case, we can rephrase the question as follows: \emph{Given a graph with $n$ nodes and diameter $k$, what is the smallest possible number of edges $m$?} Clearly, such a graph would maximize social welfare restricted to the set of strongly connected graphs; we can extend this result to all graphs.

\begin{figure}[t]
\begin{center}
\includegraphics[width = 2.5in]{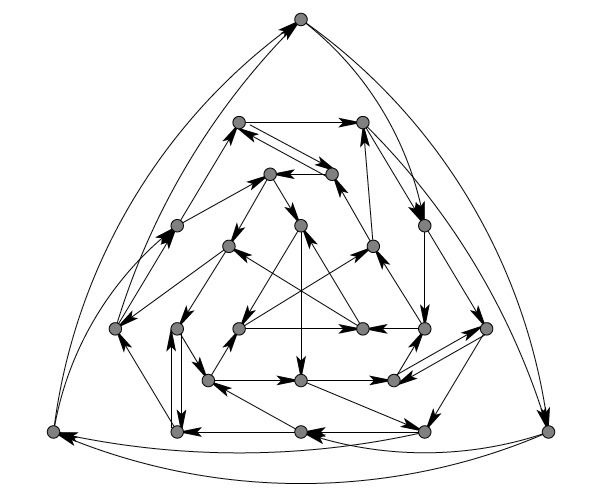}
\caption{Kautz graph with $24$ vertices, outdegree $2$ and diameter $4$. This is a symmetric-efficient graph for $n = 24$, $k = 4$, and $c < 10$, and is stable for $c < 1$.}
\label{fig:kautz}
\end{center}
\end{figure}

\begin{thm}
\label{thm:kautz}
For any $k \geq 4$, $n \geq 16$, and $c < \xfrac{n}{\log_k(n)}$ the  Kautz graph $K^{k-1}_{\log_k(n)}$ (see Figure~\ref{fig:kautz}) is asymptotically symmetric-efficient\footnote{Generally, we assume that $\log_k(n) \in \mathbb Z$, and then can make a statement about asymptotic as $n \to \infty$ for a fixed $k$ on the set of well-defined $n$.} and its social welfare is 
{\[ n ((n-1) - c{\log_k (n)}) . \]}
Moreover, for $c \leq 1$ the Kautz graph is pairwise stable.
\end{thm}

\begin{proof}[of Theorem~\ref{thm:kautz}]
The Kautz graph is strongly connected and all vertices have the same degree; hence it is symmetric. Each vertex is connected to all others and each has degree $d_n = {\log_k(n)}$, hence each vertex contributes {$(n-1)-c\log_k(n)$} to the social welfare. 

As discussed in the main body of the text, the Kautz graph is known to asymptotically have the fewest number of edges necessary in order to have a graph of size $n$ with maximum diameter $k$ (see \cite{MS2005} for a survey), and hence it is asymptotically optimal amongst the set of strongly connected graphs. Analogously to the Proof of Theorem~\ref{thm:flower} (see the last paragraph), we can bootstrap this result in order to prove efficiency amongst the set of all graphs by showing that the social welfare of a graph that is not strongly connected is (asymptotically) at most the social welfare of the corresponding Kautz graph for this range of parameters. 

We prove the contrapositive. Assume there is a graph that is not strongly connected; take two strongly connected components of size $a\geq 1$ and $b = n-a$; we show that the Kautz graph of size $n$ has at least as much social welfare as this graph. Clearly, from above, the social welfare of each component can only improve by making it a Kautz graph, and this does not affect the social welfare of the remaining graph. Then, (since the two components may be connected) the social welfare is at most 
\begin{eqnarray*}
a(a-1) &-& cad_a + (n-a)(n-1) - c(n-a)d_{n-a} \\
&&= n(n-1) - a(n-a) - c\left(\log_k(a^a{(n-a)}^{n-a})\right) \\
&&\leq n(n-1) - (n-1) - c\left(\log_k({(n-1)}^{n-1})\right) \\
&&\leq n(n-1) - c\log_k({n}^{n})  + c\log_k\left(\frac{{n}^{n}}{{(n-1)}^{n-1}}\right) - (n-1) \\
&&\leq n(n-1) - c\log_k({n}^{n}) 
\end{eqnarray*}
where the last inequality follows for any constant $c$ and $k$ and large enough $n$. This is the social welfare of the Kautz graph; hence the Kautz graph is asymptotically symmetric-efficient.

We now prove stability. By Lemma \ref{pairwise_stable_if_stable_and_scc}, this also gives pairwise stability. 
Clearly, no strategy $S_x'$ such that $|S_x'| \geq |S_x|$ can improve the utility of $x$. Hence we first consider strategies $S_x' \subseteq S_x$. Let $x = x_0, \ldots, x_D$ and $y = y_0, \ldots, y_D$ be two nodes such that there is there is no $i$ such that $x_i = y_{D-i-1}, \ldots, x_D = y_0, \ldots, y_{i}$. Hence, the distance from $x$ to $y$ must be exactly $k$ in the Kautz graph. Consider any edge $xz$ such that $z_0 \neq y_0$. It must again hold that there is no $i$ such that $z_i = y_{D-i-1}, \ldots, z_D = y_0, \ldots, y_i$; otherwise there would have been some such $i$ for $x$. Hence, the distance from $x$ to $y$ is $k+1$ if we remove  edge $xx'$ where $x' = x_1,\ldots,x_D,y_0$. Hence, the strategy $S_x \setminus xx'$ does not improve $x$'s utility since the cost reduces by at most $1$ while at least one node becomes unreachable. 
Note that, for every allowable $y_0$, some such vertex $y$ exists. Hence, no strategy $S_x' \subseteq S_x$ improves $x$'s utility. Lastly, for sake of contradiction, assume there is an optimal strategy $S_x'$ such that $|S_x'| < |S_x|$. Then, there must be some allowable $y_0$ to which $x$ does not build an edge. However, as we just observed, $S_x' + xy_0$ improves the utility of $x$. This contradicts the optimality of $S_x'$. Hence, this concludes the proof.
\end{proof}

\begin{rem}
We leave open the question of whether Kautz graphs are stable for any $c > 1$, and more generally, the question of characterizing  optimal symmetric stable graphs. 
\end{rem}

\subsection{Price of Stability and Price of Anarchy}
\label{sec:extremal}
We can further observe some properties of the extremal ranges of the parameters where the stable and efficient graph structures are not as interesting. These results also lead to the following observation.

\begin{cor}
The price of anarchy is 0 and the price of stability is $1$.
\end{cor}
This follows from the definition of the price of anarchy and price of stability by combining Propositions \ref{prop:empty} and~\ref{efficient cycle} for the former and Propositions~\ref{stable cycle} and \ref{efficient cycle} for the latter.

\begin{prop}
\label{prop:empty}
If $c_s \geq 1$ or $c_\ell \geq 1$ then the empty graph is pairwise stable. Moreover, if $c_s > n-1$ or $c_\ell > n-1$ then the only stable graph (and hence the only pairwise stable graph) is the empty graph. 
\end{prop}
\begin{proof}[of Proposition~\ref{prop:empty}]
Assume without loss of generality that $c_s\geq1$. We will prove the contrapositive. Assume that $G$ is not pairwise stable. Since the graph is empty, there are no possible edges to remove. Thus, there must exists a pair of vertices  $v,w$ such that $U^s_{G+s_{vw}+\ell_{wv}}(v) > U^s_G(v)$ and $U^\ell_{G+s_{vw}+\ell_{wv}}(w) \geq U^\ell_G(w)$. This implies that $U^s_{G+s_{vw}+\ell_{vw}}(v) =U^s_G(v)+1-c_s > U^s_G(v)$, and thus $c_s < 1$. 

Now, suppose that $c_s >  n-1$. Note that in a stable network, the utility for any vertex $v$ is at least $0$ since this can always by attained by taking the strategy $S_v = \emptyset$. Assume that a graph $G$ is nonempty. Thus, there must exist at least one vertex $v$ such that $ds^+(v) > 0$. Let $v$ be such a vertex. Then $U^s_G(v) = |R^s_G(v)| - c_s\cdot ds^+(v) \leq (n-1) - c_s < 0$. Thus, $G$ is not stable.
\end{proof}


\begin{prop}
\label{efficient cycle}
If $k = \infty$, and $0 < c_s \leq n-1, 0 < c_\ell \leq n-1$, then the cycle is the only efficient network. 
\end{prop}
\begin{proof}
For ease of notation, we drop $k = \infty$ from the subscripts.
For an arbitrary node $v$, if $ds^+_G(v)=0$ then there does not exist a speaking path from $v$ to rest of the network, and hence  $U_G^s(v) = \vert R^s_G(v) \vert - c_s \cdot ds_G^+(v) = 0$. 
Otherwise, if $ds_G^+(v) \geq 1$, then the rest of network may be reachable from $v$. Hence,  $\vert R^s_G(v) \vert \leq n-1$ where equality occurs if and only if all other nodes are reachable from $v$.
Therefore 
\[U^s_G(v) = \vert R^s_G(v) \vert - c_s \cdot ds_G^+(v) \leq (n-1)- c_s ,\] 
and we know that $c_s<(n-1)$. Hence, in the case of equality, the utility is improved over the setting where $ds_G^+(v) = 0$.
Similarly, it follows that \[U^\ell_G(v)  \leq (n-1)- c_\ell.\] Thus we can conclude that \[U_G(v) = U^s_G(v) + U^\ell_G(v) \leq (n-1)-(c_s+c_\ell) \] so
\[\varphi(G) = \sum\nolimits_v U_G(v)\leq n\cdot((n-1)-(c_s+c_\ell)) \]
By the definition of efficiency, a network is efficient if it maximizes $\varphi$ thus if we be able to find a network say $G$ such that $\varphi(G)=n\cdot((n-1)-(c_s+c_\ell))$ then we could conclude that $G$ is efficient.

Now we consider cycle $C$ to compute $\varphi(C)$. A cycle $C$ is a closed path $v_1,v_2, \ldots v_n,v_{n+1}$ such that $ v_{n+1}=v_1,$ for all ${1\leq i\leq n} \ s_{v_iv_{i+1}} \in E_s , \ \ell_{v_{i+1}v_i} \in E_\ell$. 
Thus by the definition of reachability for any $v, w \in C, v \in R^s_{G,{\infty}}(w), v \in R^\ell_{G,\infty}(w)$ so $\vert R^s_G(v) \vert=\vert R^\ell_{G,\infty}(v) \vert=n-1$.
Moreover for any node $v \in C$ we have $ds_C^+(v)=d\ell_C^+(v)=1$. 
Therefore 
\[\varphi(G) = \sum\nolimits_v U_G(v) = \sum\nolimits_v (n-1)-(c_s+c_\ell)) = n\cdot((n-1)-(c_s+c_\ell)) .\]

Now we prove that cycle is the unique efficient network. By contradiction assume that there exists another efficient network $G \neq C$ where $\varphi(G)=n\cdot((n-1)-(c_s+c_\ell)) $. According to the first part of proof, the only condition for holding this equality is that for all ${v\in V}$ it holds that $ds_G^+(v)=d\ell_G^+(v)=1, \   \vert R^s_G(v) \vert=\vert R^\ell_G(v) \vert=n-1$, thus $G$ is connected in other words, for all ${v,w\in V}$, we have that $v \in R_{G}^s(w)$ and $v \in R_{G}^\ell(w)=1$.  Moreover since the speaking outdegree of each node is one, so if we only consider speaking edges, we will have one cycle on the nodes. Now we want to show that listening edges are exactly in opposite direction of speaking edges. If not, there must exist pair $(v,w)$ such that $s_{vw}\in E_s, \ell_{wv}\notin E_\ell$ thus by the same proof used in proposition \ref{prop:bi_stable}, node $v$ could improve her utility by removing $s_{vw}$ and this action dos not change other nodes' utility. This improvement in social welfare contradicts with efficiency of $G$.      
\end{proof}

\begin{prop}
\label{stable cycle} For  $k = \infty$,  if $0 \leq c_s \leq n-1$ and $0 \leq c_\ell \leq n-1$, then the cycle $C$ is pairwise stable.
\end{prop}
\begin{proof}
We first prove stability, and then use Lemma~\ref{pairwise_stable_if_stable_and_scc} to show that pairwise stability must also hold. Using the same proof presented in Proposition \ref{efficient cycle}, we know that for all $ {v\in V}$, we have $\vert R^s_G(v) \vert=\vert R^\ell_G(v) \vert=n-1$ and $ds_C^+(v)=d\ell_C^+(v)=1$. To prove the stability of $C$ we must show that for all agents  $v$ and any alternate strategy $S'_v \subseteq V$ , we have 
\[ U_{S_v,\mathbf S_{-v}}(v) \geq U_{S'_v,\mathbf S_{-v}}(v) .\]
There are only four possible moves for node $v_i$.
\begin{enumerate}
\item
Adding $s_{v_i w}$ such that $w\neq v_{i+1}$
\item
Removing $s_{v_iv_{i+1}}$
\item
Adding $\ell_{v_i w}$ such that $w\neq v_{i+1}$
\item
Removing $\ell_{v_iv_{i+1}}$.

\end{enumerate}
Now we show that after any such change, the utility of $v_i$ will decrease. Without loss of generality, consider a change to a speaking edge. We know that $\vert R^s_C(v) \vert=\vert R^\ell_C(v) \vert=n-1$ which is the maximum possible, so by adding a new speaking edge we could not gain;  however we will loose $c_s$ thus $U_{C+s_{v_iw}}(v_i)=U_C(v_i)-c_s$ and since $c_s>0$ so $U_{C+s_{v_iw}}(v_i)<U_C(v_i)$. In the case of deleting, if we remove a speaking edge, then $ds_C^+(v_i)=0$ thus for all ${w\in V ,w \neq v_i}$ we have $w \notin R^s_C(v_i)$ so $\vert R^s_C(v_i) \vert=0$. Therefore, $U_{C-s_{v_iw}}(v_i)=U_C(v_i)+c_s-(n-1)$, and since $c_s<(n-1)$ so $U_{C-s_{v_iw}}(v_i)<U_C(v_i)$. The same reasoning used for listening edges. Hence the cycle is stable. By Lemma~\ref{pairwise_stable_if_stable_and_scc}, it is also pairwise stable.
\end{proof}

%

We conclude this section with a characterization of all efficient and stable networks for $k=1$. The proof follows directly from the definitions.
\begin{prop}
\label{prop:complete}
Assume $k = 1$. If $c_s, c_\ell > 1$ then the empty graph is efficient and is the only stable graph. If  $c_s, c_\ell \leq 1$ then the complete graph is efficient; if $c_s < 1$ or $c_s < 1$ then the complete graph is the only only stable graph. If $c_s = c_\ell = 1$, then all graphs are efficient and stable.
\end{prop}

\section{The Generalized Clustering Coefficient}
\label{sec:clustering}

There is a wide literature on networks and their social and economic properties (see \cite{EK2010}).
A clustering coefficient is a measure of the degree to which nodes in a graph tend to be highly interconnected in small clusters. 
While clustering is high in many real-world networks, evidence suggests that social networks in particular are highly clustered \cite{HL1971,WaSt1998}. 
Many (1-dimensional) measures of clustering have been defined  \cite{LP1949,WaSt1998,OP2009,Opsa2013}. 
These clustering coefficients capture statistics about local connections; i.e., vertices that are apart by at most two edges. 
Our $k$-dimensional clustering coefficient generalizes these approaches by considering clustering between vertices that are at most $k+1$ edges apart.

Inspired by our model, we define a generalized clustering coefficient below.
We let $f(G, c, k) \in [0,1]$ be the percentage of removable edges in $G$ for the model specified by $c \in \R_+, k \in \{2, 3, \ldots, \infty\}$. 
We can define an analogous local version $f(v, c, k) \in [0,1]$ as the percentage of removable edges going out from $v$. 
These functions implicitly define clustering coefficients as a function of our model parameters. 
In particular, we define two specific such clustering coefficients which we believe are of general interest. 
\begin{definition}
\label{def:gencc}
The \emph{generalized (global) clustering coefficient of dimension $k$} of a graph $G$ to be the $k$-dimensional vector $\mathbf g$ where
\[ g_i(G) = f(G, 1, i+1, T(\cdot) = V).  \]
Similarly, the \emph{generalized (local) clustering coefficient of dimension $k$} of a vertex $v$ to be the $k$-dimensional vector $\mathbf h$ where
\[ h_i(G) = f(v, 1, i+1, V).  \]
\end{definition}

For any $c \leq 1$, an edge $uv$ is removable \emph{if and only if} there is a path of length at most $k$ from $u$ to $v$ that does not use $uv$. Hence, the 1-dimensional generalized clustering coefficient makes a statement about the number of edges $uv$ that participate in at least one closed triangle. 

\begin{rem}
Our definition of clustering is closely related to the class of clustering coefficients which defined to be the percentage of (potential) triangles that are closed by an edge \cite{WaSt1998}. We instead measure the percentage of edges that close at least one triangle. Similarly, we can generalize the original definition to higher $k$: it would measure the percentage of (potential) cycles of length at most $k$ that are closed, while we measure the percentage of edges completing at least one cycle of length at most $k$.
\end{rem}


\section{Discussion \& Future Work}
\label{sec:conclusion}


For the directed setting $(c_\ell = 0)$, we show that for $k = \infty$, asynchronous stochastic edge dynamics converge to a stable network; moreover these fixed points can have non-trivial network structure. Our proof does not generalize to the case of $k < \infty$, and we leave open the question of whether these dynamics converge in that setting. Proving bounds on the time to convergence and understanding the regions of attraction would also be of interest as they inform the distribution of networks we would expect to see from a generative version of this model.

For the bidirected setting, we show that for arbitrary $k$ the edge dynamics exhibit fast converge. Moreover, vertex dynamics (in which multiple addable or removable edges are modified) are also fast to converge. However, developing an understanding of asynchronous best-response dynamics remains open as an interesting line for future work; our results do not generalize as there is no reason why a best response would be limited to only changing addable and removable edges (given two addable edges, it is not necessarily optimal to add both).  

With respect to the static game, we give classes of efficient networks (for $k \lesssim \sqrt n$) and symmetric-efficient networks (for $4 \leq k$). While the former are stable for any $c_s, c_\ell \lesssim \xfrac k 2$, the latter are only stable for $c_s, c_\ell \leq 1$; for the latter we rely on a long line of work from combinatorics, and determining whether any symmetric stable network exists appears to be a deep and technically challenging open problem.

Lastly, a natural generalization of our model is to define target sets $T^s(v), T^\ell(v) \subseteq V$ and consider utilities
\[ U^s_{G,k}(v)  = \vert R^s_{G,k}(v) \cap T^s(v) \vert - c_s \cdot  ds_G^+(v)  \quad \mbox{ and } \quad 
 U^\ell_{G,k}(v)  = \vert R^\ell_{G,k}(v) \cap T^\ell(v) \vert - c_\ell \cdot  d\ell_G^+(v) . \]
Such a model captures agents who participate in a larger network, but are only interested in connecting to some subset of nodes. Understanding the stability and efficiency of networks under certain classes of target sets would be of interest.



\bibliographystyle{ACM-Reference-Format-Journals}
\bibliography{socialmedia}


\begin{thebibliography}{00}


\ifx \showCODEN    \undefined \def \showCODEN     #1{\unskip}     \fi
\ifx \showDOI      \undefined \def \showDOI       #1{{\tt DOI:}\penalty0{#1}\ }
  \fi
\ifx \showISBNx    \undefined \def \showISBNx     #1{\unskip}     \fi
\ifx \showISBNxiii \undefined \def \showISBNxiii  #1{\unskip}     \fi
\ifx \showISSN     \undefined \def \showISSN      #1{\unskip}     \fi
\ifx \showLCCN     \undefined \def \showLCCN      #1{\unskip}     \fi
\ifx \shownote     \undefined \def \shownote      #1{#1}          \fi
\ifx \showarticletitle \undefined \def \showarticletitle #1{#1}   \fi
\ifx \showURL      \undefined \def \showURL       #1{#1}          \fi

\bibitem[\protect\citeauthoryear{Aumann and Myerson}{Aumann and
  Myerson}{1988}]%
        {AuMy1988}
{R. Aumann} {and} {R. Myerson}. 1988.
\newblock \showarticletitle{An application of the Shapley value}.
\newblock In {\em Endogenous formation of links between players and
  coalitions}. Cambridge Univ. Press, Cambridge, UK.
\newblock


\bibitem[\protect\citeauthoryear{Bala and Goyal}{Bala and Goyal}{2000}]%
        {BaGo2000}
{V. Bala} {and} {S. Goyal}. 2000.
\newblock \showarticletitle{Self-organization in communication networks}.
\newblock {\em Econometrica\/} (2000).
\newblock


\bibitem[\protect\citeauthoryear{Barab{\'a}si and Albert}{Barab{\'a}si and
  Albert}{1999}]%
        {BA1999}
{Albert-L{\'a}szl{\'o} Barab{\'a}si} {and} {R{\'e}ka Albert}. 1999.
\newblock \showarticletitle{Emergence of scaling in random networks}.
\newblock {\em science\/} {286}, 5439 (1999), 509--512.
\newblock


\bibitem[\protect\citeauthoryear{Bollob\'as}{Bollob\'as}{2001}]%
        {Boll2001}
{B\'ella Bollob\'as}. 2001.
\newblock {\em Random Graphs}.
\newblock Cambridge University Press: Cambridge, UK.
\newblock


\bibitem[\protect\citeauthoryear{Carrington, Scott, and Wasserman}{Carrington
  et~al\mbox{.}}{2005}]%
        {csw2005}
{Peter~J Carrington}, {John Scott}, {and} {Stanley Wasserman}. 2005.
\newblock {\em Models and methods in social network analysis}. Vol.~28.
\newblock Cambridge university press.
\newblock


\bibitem[\protect\citeauthoryear{Dutta and Mutuswami}{Dutta and
  Mutuswami}{1997}]%
        {DuMu1997}
{Bhaskar Dutta} {and} {Suresh Mutuswami}. 1997.
\newblock \showarticletitle{Stable networks}.
\newblock {\em Journal of Economic Theory\/} {76}, 2 (1997), 322--344.
\newblock


\bibitem[\protect\citeauthoryear{Easley and Kleinberg}{Easley and
  Kleinberg}{2010}]%
        {EK2010}
{David Easley} {and} {Jon Kleinberg}. 2010.
\newblock {\em Networks, crowds, and markets: Reasoning about a highly
  connected world}.
\newblock Cambridge University Press.
\newblock


\bibitem[\protect\citeauthoryear{Elspas, Kautz, and Turner}{Elspas
  et~al\mbox{.}}{1968}]%
        {BKT1968}
{Bernard Elspas}, {William~H Kautz}, {and} {James Turner}. 1968.
\newblock {\em THEORY OF CELLULAR LOGIC NETWORKS AND MACHINES.}
\newblock {T}echnical {R}eport. DTIC Document.
\newblock


\bibitem[\protect\citeauthoryear{Erdos and R{\'e}nyi}{Erdos and
  R{\'e}nyi}{1960}]%
        {ErRe1960}
{Paul Erdos} {and} {Alfr{\'e}d R{\'e}nyi}. 1960.
\newblock \showarticletitle{On the evolution of random graphs}.
\newblock {\em Bulletin of the International Statistical Institute\/} {38}, 4
  (1960), 343--347.
\newblock


\bibitem[\protect\citeauthoryear{Hirschberg and Wong}{Hirschberg and
  Wong}{1979}]%
        {HiWo1979}
{D.~S. Hirschberg} {and} {C.~K. Wong}. 1979.
\newblock \showarticletitle{Upper and Lower Bounds for Graph-Diameter Problems
  with Application to Record Allocation}.
\newblock {\em Journal of Combinatorial Theory\/} (1979).
\newblock


\bibitem[\protect\citeauthoryear{Holland and Leinhardt}{Holland and
  Leinhardt}{1971}]%
        {HL1971}
{Paul~W Holland} {and} {Samuel Leinhardt}. 1971.
\newblock \showarticletitle{Transitivity in structural models of small groups.}
\newblock {\em Comparative Group Studies\/} (1971).
\newblock


\bibitem[\protect\citeauthoryear{Imase and Itoh}{Imase and Itoh}{1983}]%
        {II1983}
{Makoto Imase} {and} {Masaki Itoh}. 1983.
\newblock \showarticletitle{A design for directed graphs with minimum
  diameter}.
\newblock {\it IEEE Trans. Comput.} 8 (1983), 782--784.
\newblock


\bibitem[\protect\citeauthoryear{Jackson}{Jackson}{2003}]%
        {Jack2001}
{Matthew~O Jackson}. 2003.
\newblock \showarticletitle{The stability and efficiency of economic and social
  networks}.
\newblock In {\em Networks and Groups}. Springer, 99--140.
\newblock


\bibitem[\protect\citeauthoryear{Jackson}{Jackson}{2005}]%
        {Jack2005}
{Matthew~O Jackson}. 2005.
\newblock \showarticletitle{A survey of network formation models: stability and
  efficiency}.
\newblock {\em Group Formation in Economics: Networks, Clubs, and Coalitions\/}
  (2005), 11--49.
\newblock


\bibitem[\protect\citeauthoryear{Jackson and Watts}{Jackson and Watts}{2002}]%
        {JaWa2002}
{Matthew~O Jackson} {and} {Alison Watts}. 2002.
\newblock \showarticletitle{The evolution of social and economic networks}.
\newblock {\em Journal of Economic Theory\/} {106}, 2 (2002), 265--295.
\newblock


\bibitem[\protect\citeauthoryear{Jackson and Wolinsky}{Jackson and
  Wolinsky}{1996}]%
        {JaWo1996}
{Matthew~O Jackson} {and} {Asher Wolinsky}. 1996.
\newblock \showarticletitle{A strategic model of social and economic networks}.
\newblock {\em Journal of economic theory\/} {71}, 1 (1996), 44--74.
\newblock


\bibitem[\protect\citeauthoryear{Karandikar, Java, Joshi, Finin, Yesha, and
  Yesha}{Karandikar et~al\mbox{.}}{2008}]%
        {KJJFYY2008}
{Amit Karandikar}, {Akshay Java}, {Anupam Joshi}, {Tim Finin}, {Yaacov Yesha},
  {and} {Yelena Yesha}. 2008.
\newblock \showarticletitle{Second Space: A Generative Model for the
  Blogosphere.}. In {\em ICWSM}.
\newblock


\bibitem[\protect\citeauthoryear{Luce and Perry}{Luce and Perry}{1949}]%
        {LP1949}
{R~Duncan Luce} {and} {Albert~D Perry}. 1949.
\newblock \showarticletitle{A method of matrix analysis of group structure}.
\newblock {\em Psychometrika\/} {14}, 2 (1949), 95--116.
\newblock


\bibitem[\protect\citeauthoryear{Miller and {\v{S}}ir{\'a}n}{Miller and
  {\v{S}}ir{\'a}n}{2005}]%
        {MS2005}
{Mirka Miller} {and} {Jozef {\v{S}}ir{\'a}n}. 2005.
\newblock \showarticletitle{Moore graphs and beyond: A survey of the
  degree/diameter problem}.
\newblock {\em Electronic Journal of Combinatorics\/}  {61} (2005), 1--63.
\newblock


\bibitem[\protect\citeauthoryear{Myerson}{Myerson}{1977}]%
        {Myer1977}
{Roger~B. Myerson}. 1977.
\newblock \showarticletitle{Graphs and cooperation in games}.
\newblock {\em Mathematics of Operations Research\/} (1977).
\newblock


\bibitem[\protect\citeauthoryear{Newman}{Newman}{2003}]%
        {Newm2003}
{Mark~EJ Newman}. 2003.
\newblock \showarticletitle{The structure and function of complex networks}.
\newblock {\em SIAM review\/} {45}, 2 (2003), 167--256.
\newblock


\bibitem[\protect\citeauthoryear{Nisan, Roughgarden, Tardos, and
  Vazirani}{Nisan et~al\mbox{.}}{2007}]%
        {AGTbook}
{Noam Nisan}, {Tim Roughgarden}, {Eva Tardos}, {and} {Vijay~V. Vazirani}. 2007.
\newblock {\em Algorithmic Game Theory}.
\newblock Cambridge University Press. 411--443 pages.
\newblock


\bibitem[\protect\citeauthoryear{Opsahl}{Opsahl}{2013}]%
        {Opsa2013}
{Tore Opsahl}. 2013.
\newblock \showarticletitle{Triadic closure in two-mode networks: Redefining
  the global and local clustering coefficients}.
\newblock {\em Social Networks\/} {35}, 2 (2013), 159--167.
\newblock


\bibitem[\protect\citeauthoryear{Opsahl and Panzarasa}{Opsahl and
  Panzarasa}{2009}]%
        {OP2009}
{Tore Opsahl} {and} {Pietro Panzarasa}. 2009.
\newblock \showarticletitle{Clustering in weighted networks}.
\newblock {\em Social networks\/} {31}, 2 (2009), 155--163.
\newblock


\bibitem[\protect\citeauthoryear{Wasserman and Faust}{Wasserman and
  Faust}{1994}]%
        {wf1994}
{Stanley Wasserman} {and} {Katherine Faust}. 1994.
\newblock {\em Social network analysis: Methods and applications}. Vol.~8.
\newblock Cambridge university press.
\newblock


\bibitem[\protect\citeauthoryear{Watts}{Watts}{2003}]%
        {Watt1999}
{Alison Watts}. 2003.
\newblock \showarticletitle{A dynamic model of network formation}.
\newblock In {\em Networks and Groups}. Springer, 337--345.
\newblock


\bibitem[\protect\citeauthoryear{Watts and Strogatz}{Watts and
  Strogatz}{1998}]%
        {WaSt1998}
{Duncan~J Watts} {and} {Steven~H Strogatz}. 1998.
\newblock \showarticletitle{Collective dynamics of ‘small-world’networks}.
\newblock {\em nature\/} {393}, 6684 (1998), 440--442.
\newblock


\end{thebibliography}

\end{document}